\documentclass[graybox]{svmult}

\usepackage{mathptmx}       
\usepackage{helvet}         
\usepackage{courier}        
\usepackage{type1cm}        
\usepackage{makeidx}         
\usepackage{graphicx}        
\usepackage{multicol}        
\usepackage[bottom]{footmisc}

\usepackage{amsmath}
\usepackage{amsfonts}
\usepackage{amssymb}

\usepackage{color}

\makeindex             

\begin{document}

\title*{From bee species aggregation to models of disease avoidance: The \emph{Ben-Hur} effect}
\author{K. E. Yong, E. D\'iaz Herrera, C. Castillo-Chavez}
\institute{K. E. Yong \at Mathematics/Science Subdivision, University of Hawai`i - West O`ahu, Kapolei, HI 96707, USA, \email{kamuela.yong@hawaii.edu}
\and E. D\'iaz Herrera \at Instituto Nacional De Salud P\'ublica, Universidad  655, Santa Mar\'ia Ahuacatitl\'an, 62100, Cuernavaca, Morelos. M\'exico \email{edgar.diaz@insp.mx}
\and C Castillo-Chavez \at Simon A. Levin Mathematical, Computational and Modeling Science Center, Arizona State University, Tempe, AZ 85287, USA, \email{ccchavez@asu.edu}}
%
%
\maketitle

\abstract{The movie \emph{Ben-Hur} highlights the dynamics of contagion associated with leprosy, a pattern of forced aggregation driven by the emergence of symptoms and the fear of contagion. The 2014 Ebola outbreaks reaffirmed the dynamics of redistribution among symptomatic and asymptomatic or non-infected individuals as a way to avoid contagion. In this manuscript, we explore the establishment of clusters of infection via density-dependence avoidance (diffusive instability). We illustrate this possibility in two ways: using a phenomenological driven model where disease incidence is assumed to be a decreasing function of the size of the symptomatic population and with a model that accounts for the deliberate movement of individuals in response to a gradient of symptomatic infectious individuals. The results in this manuscript are preliminary but indicative of the role that behavior, here modeled in crude simplistic ways, may have on disease dynamics, particularly on the spatial redistribution of epidemiological classes.}


\section{Introduction}

The effect that aggregation of susceptible and infected populations of individuals has on the basic reproduction number, $\mathcal{R}_0$, and the \emph{final size} has been studied by various researchers (see \cite{Adler-kretzschmar1992, adler1992, Andersson1998, Bichara2015sis, Brauer2008, Fenichel2011, Morin2013}). The effect of aggregation on $\mathcal{R}_0$ and the final outbreak size is not necessarily the same as a small core group with a high activity level can substantially contribute to $\mathcal{R}_0$ while having little impact on the final outbreak size \cite{Diekmann2012}. O. Diekmann et al. \cite{Diekmann1991} showed that aggregation of susceptible and infective individuals reduces the number of groups required to capture the dynamics of a large system provided that one assumes identical levels of infectivity for all groups. These researchers also observed that increased levels of aggregation may lead to lower values of $\mathcal{R}_0$ \cite{Brauer2008, Diekmann1991}.

Spatial transmission of diseases has been studied by various researchers \cite{Riley2007, Kermack1932, Murray2002_1, Murray2002_2, Thieme2003}, often using reaction diffusion equations (see \cite{Berres2011, Carrero2003, Castillo1996, Castillo2008, Li2008, Nallaswamy1982, Sun2009, Wang2010}). In this paper, two novel reaction-diffusion models are introduced that model the spread of a communicable disease when the presence of symptoms reduces contacts among all types and, in the process, ameliorates disease spread (Model \eqref{stage}). We also examine the impact that the movement of individuals, in response to gradients of symptomatic infectious individuals modeled via cross-diffusion (Model \eqref{general system}), has on disease dynamics. This paper is organized as follows: Section \ref{Edgar'sModel} introduces a phenomenological model and identifies conditions for clustering via diffusive instability; Section \ref{cross-diffusion} examines the role of cross-diffusion on epidemiological spatial aggregation;  Section \ref{Discussion and Conclusion} collects thoughts and conclusions.

\section{Phenomenological model}\label{Edgar'sModel}
Epidemics are capable of generating shifts on population level interactions possibly as a function of the presence of growing levels of severe infection as reflected by the impact of symptomatic populations \cite{Castillo2002, Hadeler1995core, Heiderich} on the contacts between individuals and survival. A simple epidemiological model that accounts for reductions in transmission as the size of the symptomatic population increases is described below motivated by observed disease patterns in leprosy \cite{Balina1994,Robbins2009,Rodrigues2011}, Ebola \cite{Chowell2004,Kiskowski2014,Nishiura2014,Towers2014}, and influenza \cite{Rios2011}. We let $S(x,y,t)$ denote the susceptible population at time $t$ and position $(x,y)$, and divide the infected population in two groups, a group that exhibits symptoms and a group that does not, the ``asymptomatic" infectious group. Specifically, we let $I_1(x,y,t)$ denote the symptomless infectious population, assumed to be infectious, and let $I_2(x,y,t)$ denote the infected population with visible symptoms. The incidence term in a susceptible-infectious-susceptible (SIS) type model is modified by the addition of spatial diffusion to each class under the assumption that the symptomatic class, that is, $I_2$-members are in principle, to be avoided. The model equations are given by the following phenomenologically derived reaction-diffusion epidemiological model:
\begin{equation}\label{stage}
\begin{array}{rcl}
\frac{\partial S}{\partial t} & = & -\frac{\beta}{1+I_2} SI_1 +\alpha I_2 +D_S \nabla^2S, \\
\frac{\partial I_1}{\partial t} & = & \frac{\beta}{1+I_2} SI_1 - \delta I_1 + D_{I_1} \nabla^2I_, \\ 
\frac{\partial I_2}{\partial t} & = & \delta I_1 - \alpha I_2 +D_{I_2} \nabla^2I_2,
\end{array}
\end{equation}
where $\nabla^2=\Delta=\partial^2/\partial x^2+\partial^2/\partial y^2$, the Laplace operator. Setting $I_2=0$ leads to the ``standard'' $SIS$ system with diffusion \cite{Kermack1932}. The incidence term gets altered by assuming that all contacts decrease with the size of the $I_2$-population, that is, the incidence is modeled as follows:
\begin{equation}
\frac{\beta}{1+I_2} SI_1.
\end{equation}
The question posed in \cite{Diaz2010} is whether or not System \eqref{stage} can support non-uniform distributions via diffusive instability. The assumption of constant population size implies, without loss of generality, that we can take $S\equiv1-I_1-I_2$, a substitution that allows us to focus on the equations for $I_1$ and $I_2$. We observe that System \eqref{stage} supports the following positive steady states in the absence of diffusion ($D_S=D_{I_1}=D_{I_2}=0$):
\begin{equation*}
(I_1^*,I_2^*) = \left(\frac{\alpha(\beta-\delta)}{\beta \alpha+ \beta \delta + \delta^2}, \frac{\delta(\beta-\delta)}{\beta \alpha+ \beta \delta + \delta^2} \right),
\end{equation*}
from where we identify the basic reproductive number as
\begin{equation*}
\mathcal{R}_0= \frac{\beta}{\delta}.
\end{equation*}
The effects of small perturbations  of the  $(I_1^*, I_2^*)$-equilibrium are introduced via the following variables:
\begin{equation}\label{perturbations}
\ell_i(x,y,t) = I_i(x,y,t) -  I_i^* \,, \qquad i=1,2\,.
\end{equation}
Substituting \eqref{perturbations} into the last two equations of System (\ref{stage}) leads, after ignoring higher order terms, to the following linearized system
\begin{equation}\label{stagelinear}
\begin{aligned}
\frac{\partial \ell_1}{\partial t}  = & J_{11} \ell_1 + J_{12} \ell_2 + D_{I_1} \nabla ^2 \ell_1, \\
\frac{\partial \ell_2}{\partial t} =& J_{21}\ell_1 + J_{22}\ell_2 + D_{I_2} \nabla ^2 \ell_2,
\end{aligned}
\end{equation}
where the matrix $(J_{ij})$ is the Jacobian of System (\ref{stage}) in the absence of diffusion evaluated at the equilibrium $(I_1^*,I_2^*)$, namely
\begin{equation}\label{A-matrix}
J = (J_{ij}) = \left (
\begin{array}{cc}
 \frac{\alpha(\delta - \beta)}{\alpha+2\delta}  & \frac{\alpha(\alpha^2 - \beta^2)}{\beta(\alpha+2\delta)}  \\
   \delta & -\alpha   \\
\end{array}
\right ).
\end{equation}
The three conditions that guarantee diffusive instability (\cite{Segel1972}) are given by the following inequalities: 
\begin{eqnarray}
J_{11} + J_{22}  &<&0\,, \label{Turing_no_cross-diffusion1}\\
J_{11} J_{22} - J_{12} J_{21}& > & 0 \label{Turing_no_cross-diffusion2}, \\
J_{11}D_{I_2} + J_{22} D_{I_1} &> & 2\sqrt{D_{I_1}D_{I_2}(J_{11} J_{22} - J_{12} J_{21})}. \label{Turing_no_cross-diffusion3}
\end{eqnarray}
Condition \eqref{Turing_no_cross-diffusion1} always holds, since
\begin{equation*}
J_{11} + J_{22} = - \alpha \left( \frac{\delta + \beta + \alpha }{\alpha + 2 \delta}\right) < 0.
\end{equation*}
Condition \eqref{Turing_no_cross-diffusion2} is satisfied provided that
\begin{equation*}
J_{11} J_{22} - J_{12} J_{21} = \frac{\beta\alpha^2(\beta - \delta) + \delta\alpha(\alpha^2 -\beta^2)}{\beta(\alpha+2\delta)}
\end{equation*}
is positive, which is true as long as
\begin{equation*}
\beta^2(\alpha +\beta)> \alpha\delta (\beta+\alpha),
\end{equation*}
or, equivalently as long as $\mathcal{R}_0=\frac{\beta}{\delta}>1$ and $\frac{\beta}{\alpha}>1$. Now, we make use of the fact that Condition \eqref{Turing_no_cross-diffusion3} is equivalent to the inequality
\begin{equation}\label{Turing_no_cross-diffusion3_alternate}
J_{11} J_{22} - J_{12} J_{21} - \frac{1}{4D_{I_1}D_{I_2}}\left(J_{11}D_{I_2} + J_{22} D_{I_1}\right)^2 < 0 .
\end{equation}
After substituting the corresponding values from Equation \eqref{A-matrix} we see that whenever the following inequality
 \begin{equation}\label{con3}
2D_{I_1}D_{I_2}\frac{\alpha+2\delta}{\alpha \beta} \left[ \beta^2(\alpha+ 2 \delta) - \alpha \delta (\beta + 2 \alpha) \right] - D_{I_1}^2 (\delta-\beta) - D_{I_2}^2 (\alpha + 2 \delta)^2  < 0,
\end{equation}
is satisfied, Condition \eqref{Turing_no_cross-diffusion3_alternate} is satisfied. Using $\mathcal{R}_0 > 1$ leads to
\begin{equation}\label{E1}
\frac{\delta}{\beta}(\beta+2\alpha)  <  \beta + 2\delta;
\end{equation}
while $ \alpha <\beta$ leads to
\begin{equation}\label{E2}
-\frac{\beta}{\alpha}(\alpha+2\delta)  <  -(\alpha +2\delta). 
\end{equation}
The addition of Conditions \eqref{E1}-\eqref{E2} leads to the inequality
\begin{equation}
\frac{\delta}{\beta}(\beta+2\alpha) - \frac{\beta}{\alpha}(\alpha+2\delta) <\beta -\alpha .
 \end{equation}
Thus, we conclude that Condition \eqref{Turing_no_cross-diffusion3} (Inequality (\ref{con3})) holds as long as 
\begin{equation}
\mathcal{R}_0= \frac{\beta}{\delta}>1 \ \ \ \mbox{and} \ \ \ \ \frac{\beta}{\alpha}>1 \label{phenomenological instability conditions}
\end{equation}
The main conclusion of this section can be stated as follows:
\begin{theorem}
 The linear System (\ref{stagelinear}) satisfies necessary and sufficient conditions for diffusive instability whenever $\mathcal{R}_0>1$ and $\frac{\beta}{\alpha}>1$. In other words, diffusive instability takes place when the endemic state exists ($\mathcal{R}_0>1$) and $I_2$ individuals are not infectious for too long.
\end{theorem} 
The steady state non-uniform distribution of infected individuals (symptomatic and asymptomatic) loses stability due small perturbations of the form

\begin{equation}
\ell_i (x,t) = \alpha_i \cos (qx) e^{\sigma t}, \qquad i=1,2.
 \end{equation}
The present analysis works as long as the perturbations are sufficiently small to make the linear approximation (Model \eqref{stagelinear}) a valid representation of the truly nonlinear representation  
 of Model \eqref{stage}. When the perturbations have been amplified beyond a small size, the analysis is not longer adequate.
As a result of the above analysis, we expect that an initial spatially distributed population, will begin to ``break up" and aggregate according to the presence or absence of symptoms. See Figure \ref{phenomenological1}-\ref{phenomenological3} generated via the simulations carried out under Condition \eqref{phenomenological instability conditions}. We see that aggregation occurs faster if the difference in diffusion rates is large for both the linear Model \eqref{stagelinear} (see Figure \ref{phenomenological1}-\ref{phenomenological2}) and nonlinear Model \eqref{stage} (see Figure \ref{phenomenological3}). 

\begin{figure}
   \includegraphics[width=5.0in]{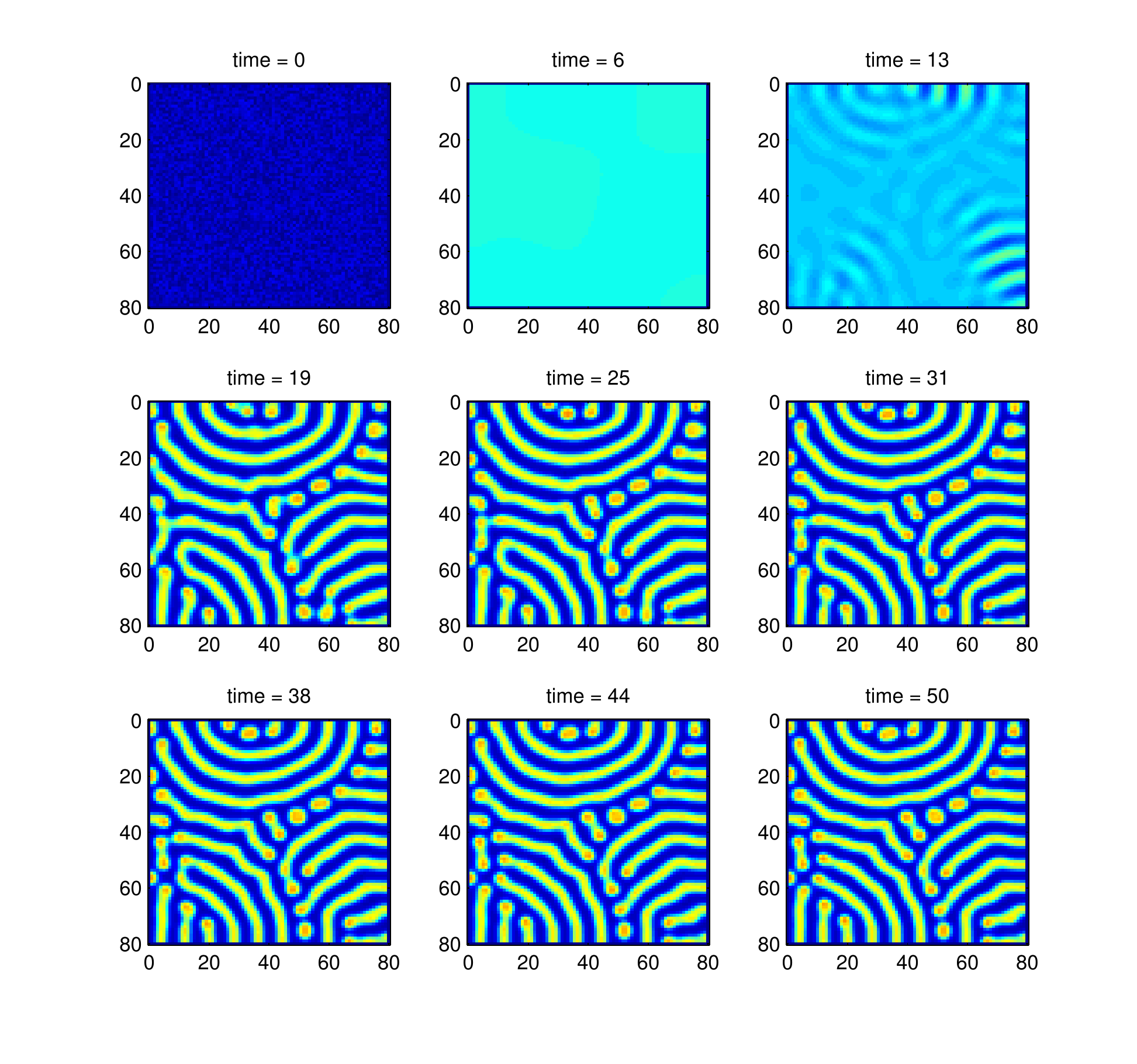}
 \caption{Spatial aggregation for Model \eqref{stagelinear} occurs quickly when the difference between diffusion rates is large.  ($D_{I_1}=10$ and $D_{I_2}=80$.  $\alpha =0.05, \beta = 0.13$ and $\delta = 1.3$, chosen so that Condition \eqref{phenomenological instability conditions} is satisfied).  \label{phenomenological1}}
\end{figure}

\begin{figure}
   \includegraphics[width=5.0in]{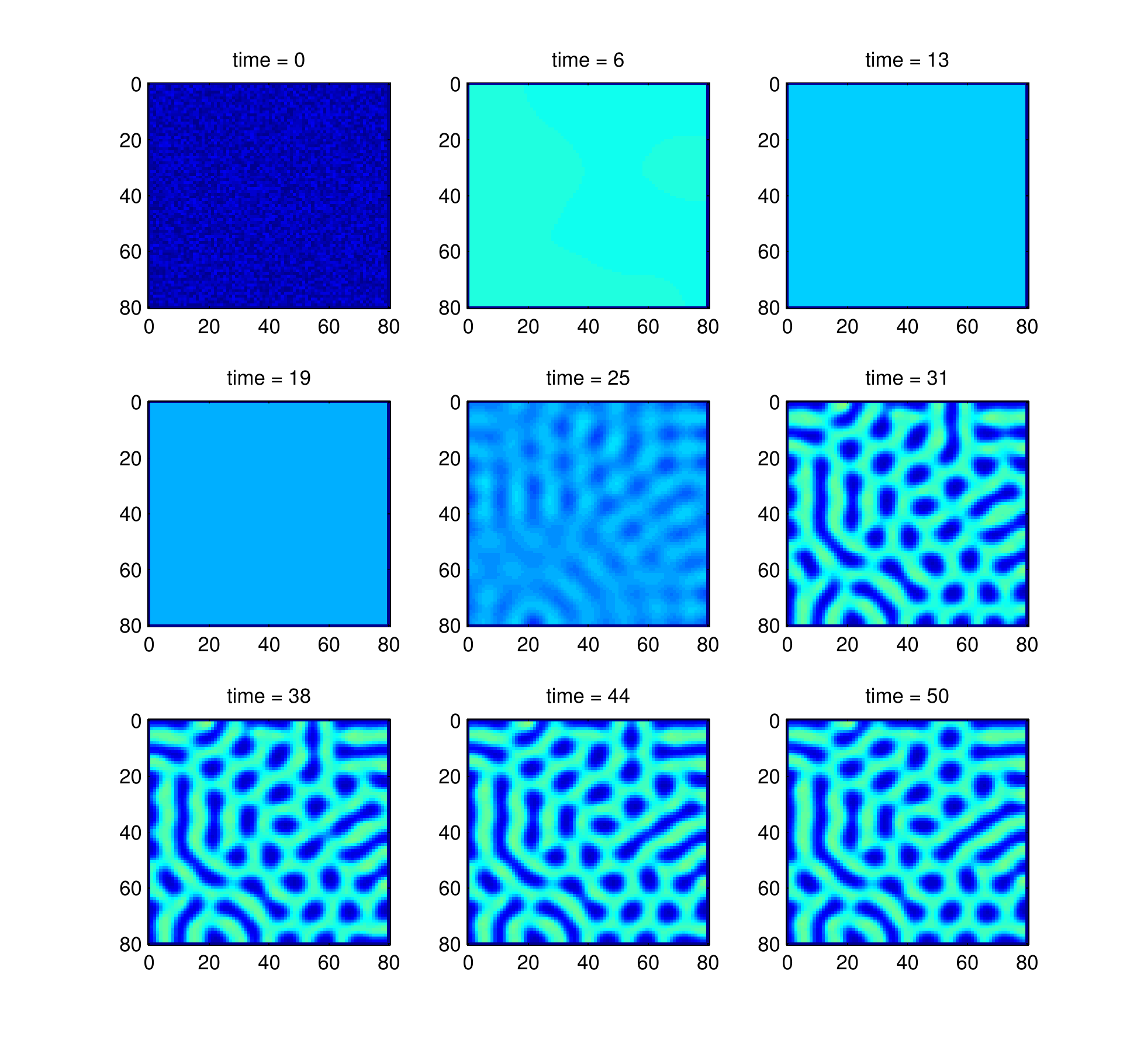}
 \caption{Spatial aggregation for Model \eqref{stagelinear} occurs slowly when the difference between diffusion rates is small.  ($D_{I_1}=10$ and $D_{I_2}=20$.  $\alpha =0.05, \beta = 0.13$ and $\delta = 1.3$, chosen so that \eqref{phenomenological instability conditions} is satisfied).   \label{phenomenological2}}
\end{figure}

\begin{figure}
 \begin{tabular}{cc}
   \includegraphics[width=3.0in]{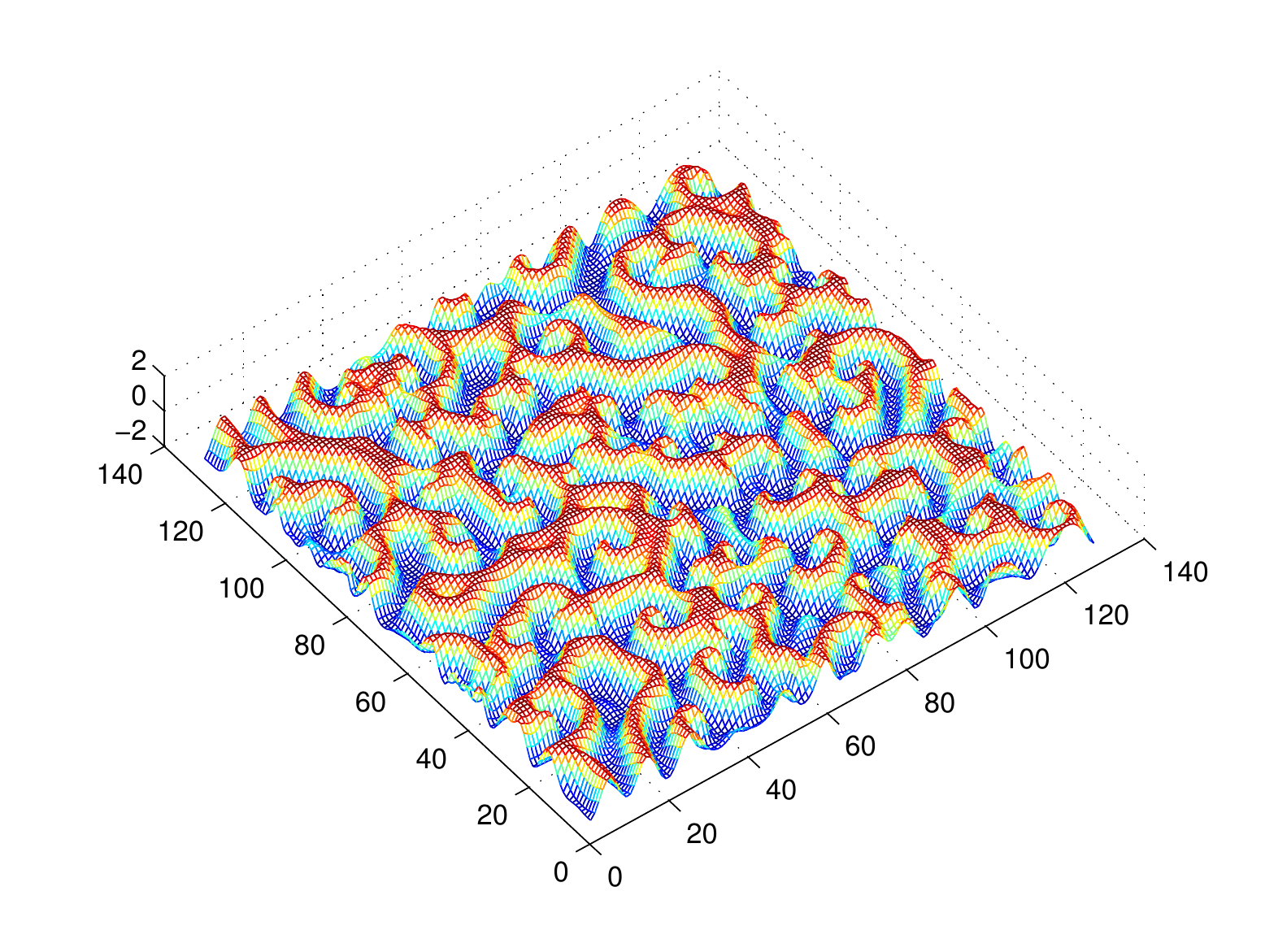}
   \includegraphics[width=3.0in]{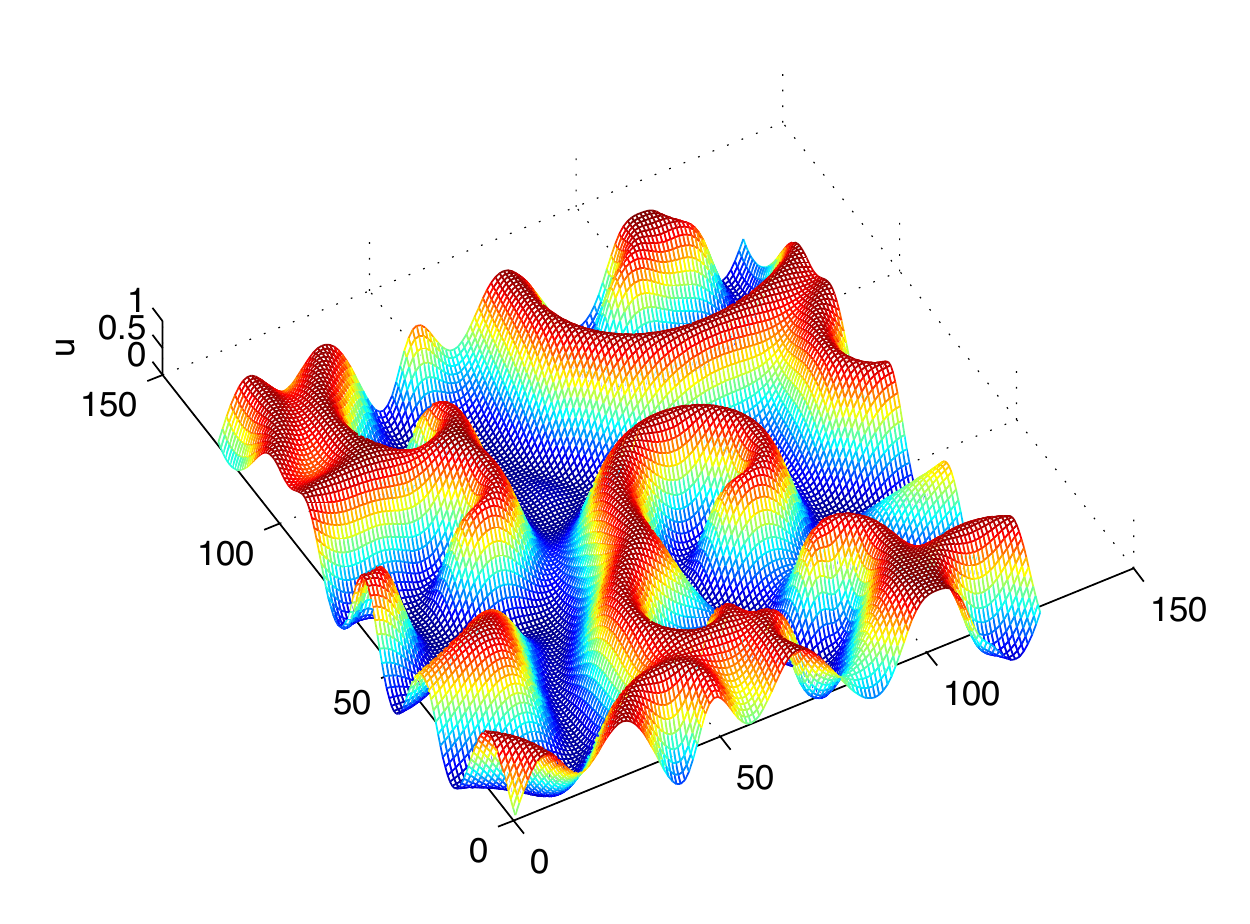}
 \end{tabular}
\caption{When the difference between diffusion rates, $D_{I_1}$ and $D_{I_2}$ from Model \eqref{stage} is large (left)
aggregation occurs faster than when the difference between diffusion rates is small (right). 
Not linear model, we use the same parameters from model (\ref{stagelinear}), except for a bigger $\beta$
( $\alpha =0.05, \beta = 1$ and $\delta = 1.3$, chosen so that \eqref{phenomenological instability conditions} is satisfied). 
\label{phenomenological3}}
\end{figure}

\section{Cross-diffusion models}\label{cross-diffusion}

The dynamics of solitary and honey bees and their role in enhancing cross-pollination in California almond tree farms was studied via a cross-diffusion model in \cite{Yong2012}. The model for the interaction of honey bees, $u_1(x,y,t)$, and solitary bees, $u_2(x,y,t)$ at time $t$ and position $(x,y)\in\Omega$, proposed in \cite{Yong2012}, is given by the system:
\begin{equation}\begin{aligned}
\frac{\partial }{\partial {t}}u_{i}&=\nabla^2\left( \alpha _{i}+\beta _{i1}u_{1}+\beta _{i2}u_{2}\right) u_{i}
+\gamma _{i}\nabla\cdot\left(u_{i}\nabla {W}\right) & \mbox{in}~\Omega\times(0,T),\\
u_i(x,y,0)&= \xi_i(x,y) & \mbox{on}~\Omega\times\{t=0\},\\
\frac{\partial{u}_i}{\partial\nu}&=0 & \mbox{on}~\partial\Omega\times(0,T),
\end{aligned}\label{SKT}
\end{equation}
where $\alpha_i\geq0$ represents the intrinsic diffusion, $\beta_{ij}\geq0$ represents the self-diffusion for $i=j$ and cross-diffusion for $i\neq j$, $W=W(x,y,t)$ represents the environmental potential, and $\gamma_i\in\mathbb{R}$ is the coefficient associated with $W$. The dynamics of avoidance between honey and solitary bees was captured by the addition of cross- and self-diffusion terms to the model in \cite{Shigesada1979}. Numerical simulations were used to show that cross-diffusion was  indeed capable of capturing the observed spatial aggregation of individuals by species. The resulting spatial aggregating of bees by species, as a result of a strong cross-diffusion ($\beta_{12})$, is illustrated in Figure \ref{bees}. This figure shows that in areas of high solitary bee density ($u_2$) result in low honey bee density ($u_1$) and in areas of low solitary bee densities result in honey bees aggregating in high densities.

\begin{figure}
 \begin{tabular}{cc}
   \includegraphics[width=3.0in]{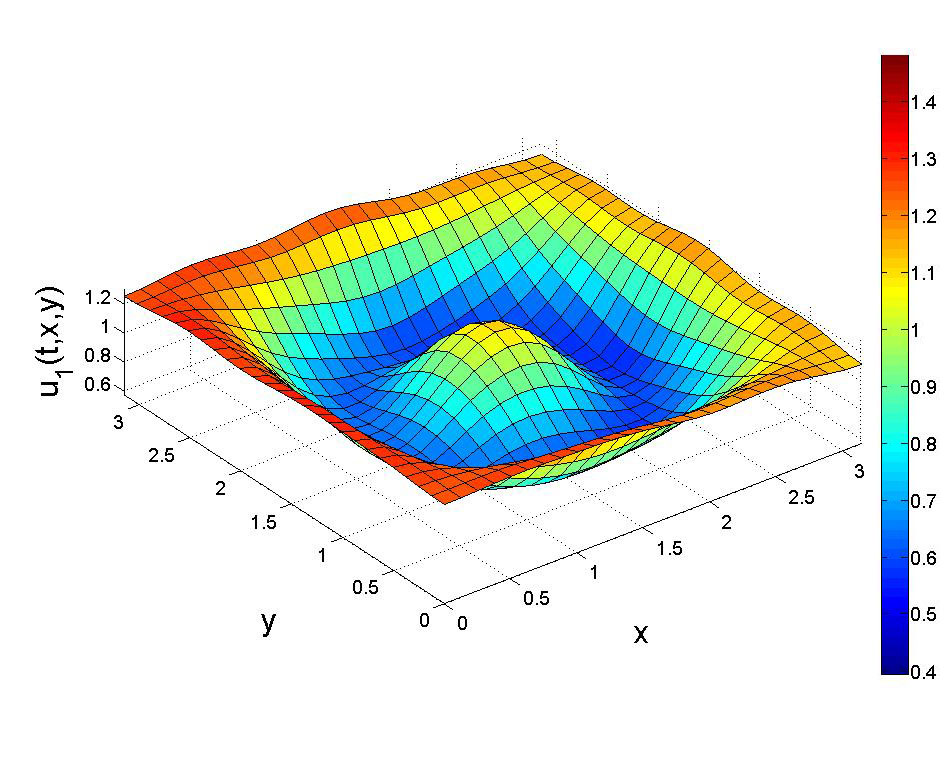}
   \includegraphics[width=3.0in]{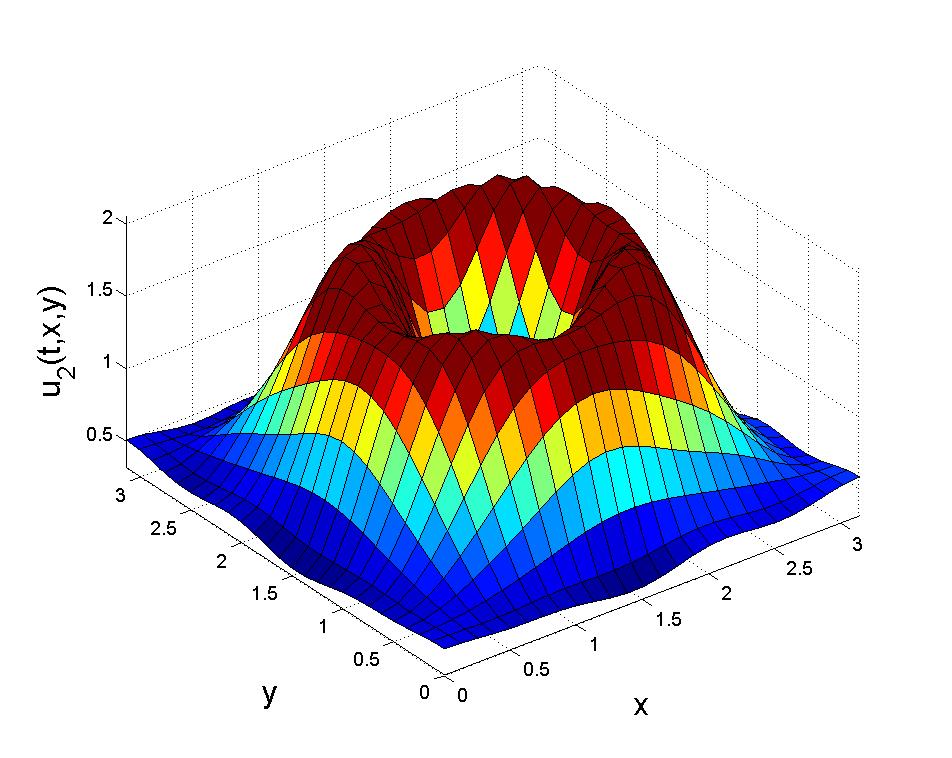}
 \end{tabular}
 \caption{The effects of a high cross-diffusion effect of solitary bees on honey bees ($\alpha_1=\alpha_2=\beta_{11}=\beta_{21}=\beta_{22}=1$, $\gamma_1=\gamma_2=5$, $\beta_{12}=10$). Honey bees ($u_1$) are in low densities in areas where solitary bees ($u_2$) are in high densities and honey bees are found in high densities in areas where solitary bees are in low densities, thus demonstrating the avoidance effects of cross-diffusion \cite{Yong2012}. \label{bees}}
\end{figure}

The use of cross-diffusion to model spatially explicit epidemics has been studied in the past (see \cite{Li2008,Nallaswamy1982,Sun2009,Wang2010}). Most recently, the role of density-dependent cross-diffusion in epidemiology has been explored numerically by Berres and Ruiz-Baier \cite{Berres2011} via the model
\begin{equation}
\begin{aligned}
\frac{\partial S}{\partial t}&=rS\left(1-\frac{S}{K}\right)-\beta \frac{SI}{S+I}+D_{S}\nabla^2 S+c\nabla\cdot(S\nabla I)\\
\frac{\partial I}{\partial t}&=\beta \frac{SI}{S+I}-\gamma I+D_{I}\nabla^2 I
\end{aligned}
\label{BerresModel}
\end{equation}
where $K$ is the carrying capacity, $r$ is the intrinsic birth rate, $\beta$ is the transmission rate, $\gamma$ is the recovery rate, $D_{S}$ and $D_{I}$ are the susceptible and infective diffusion coefficients, respectively, and $c$ is the cross-diffusion coefficient.

Following the approach in \cite{Yong2012}, the role of density-dependent cross-diffusion in the aggregation of individuals according to epidemiological states during a nefarious disease outbreak is carried out below. We expand on the type of cross-diffusion model in \cite{Bichara2015sis} via the use only of a population of $S$-individuals (susceptible) and $I$-individuals (infectives), that is, symptomatic infectious individuals. The model below assumes that symptoms generate avoidance.

\subsection{$SI$ model with diffusion}
As a starting point, let the densities for populations susceptible to a disease and infective with a disease, at time $t$ and position $(x,y)\in\Omega$ be $S(x,y,t)$ and $I(x,y,t)$, respectively. We assume the model takes the form of the following reaction-diffusion model

Incorporating the reaction and diffusion terms leads to the following $SI$ epidemiological system
\begin{equation}
\begin{aligned}
\frac{\partial S}{\partial t}&=rS^{\alpha_1}\left(1-\frac{S^{\alpha_1}}{K}\right)-\beta \frac{SI}{(S+I)^{\alpha_2}}+D_{S}\nabla^2 S+c\nabla\cdot(S\nabla I)\\
\frac{\partial I}{\partial t}&=\beta \frac{SI}{(S+I)^{\alpha_2}}-\gamma I+D_{I}\nabla^2 I.
\end{aligned}
\label{general system}
\end{equation}
Finally, it is further assumed that we have a closed system involving no external input; thus the use of Neumann boundary conditions
\begin{equation}
\frac{\partial }{\partial {\nu}}S=\frac{\partial }{\partial {\nu}}I=0,
\end{equation}
is acceptable. The initial conditions are as follows
\begin{equation}
S(x,y,0)=S_0(x,y) \qquad \mbox{and} \qquad I(x,y,0)=I_0(x,y).
\end{equation}

Whenever $D_{S}$ and $D_{I}$ are the dominant coefficients, System \eqref{general system}  reduces essentially to the heat equation, which under Neumann boundary conditions will go to the average of the initial data as $t\rightarrow\infty$ \cite{Ni1998}.

\subsection{Effects of recruitment}\label{section3}
Next we examine System \eqref{general system} with $\alpha_2=1$, that is, we focus on the study of the effects of recruitment. When $\alpha_1=1$, logistic recruitment, System \eqref{general system} reduces to System \eqref{BerresModel}. 
\begin{lemma}\label{BerresModel_Lemma}
System \eqref{BerresModel} will support Turing's diffusive instability if
\begin{equation}
\mathcal{R}_0:=\frac{\beta}{\gamma}<\frac{r}{\gamma}+1,\label{cross-diffusion cond1}
\end{equation}
\begin{equation}
Z:=-D_{S}\frac{(\beta-\gamma)}{\mathcal{R}_0}-D_{I}r+D_{I}\beta-D_{I}\frac{\gamma}{\mathcal{R}_0}-\frac{(\beta-\gamma)^2}{\beta}c\frac{K(r-(\beta-\gamma))}{r}>0,\label{cross-diffusion cond2}
\end{equation}
and
\begin{equation}
Z^2\geq D_{S}D_{I}\gamma(\beta-\gamma)(r-\beta+\gamma).\label{cross-diffusion cond3}
\end{equation}
\end{lemma}

\begin{proof}
\smartqed
To show Turing's diffusive instability, we first examine System \eqref{BerresModel} without diffusion terms ($D_{S}=D_{I}=c=0$). The corresponding endemic equilibrium point is
\begin{equation}
(S^*,I^*)=\left(\frac{K(r-\beta+\gamma)}{r},\frac{K(r-\beta+\gamma)(\beta-\gamma)}{r\gamma}\right),
\end{equation}
the basic reproduction number is
\begin{equation}
\mathcal{R}_0=\frac{\beta}{\gamma},
\end{equation} 
and the Jacobian of System \eqref{BerresModel} without diffusion evaluated at the endemic equilibrium is
\begin{equation}
J=\left(\begin{array}{cc}
-r+\beta-\frac{\gamma^2}{\beta}	&	-\frac{\gamma^2}{\beta}\\
\frac{(\beta-\gamma)^2}{\beta}	&	-\frac{\gamma(\beta-\gamma)}{\beta}
\end{array}\right).
\end{equation}
By \cite{Kumar2011} and \cite{Zemskov2013}, Turing's diffusive instability occurs if the following four conditions are satisfied
\begin{eqnarray}
&&\mbox{tr} J=J_{11}+J_{22}<0\label{condJ1}\\
&&\det J=J_{11}J_{22}-J_{12}J_{21}>0\label{condJ2}\\
&&\det\hat{D}=\hat{D}_{11}\hat{D}_{22}-\hat{D}_{12}\hat{D}_{21}>0\label{cond1}\\
&&(\hat{D}_{11}-\hat{D}_{22})^2+4\hat{D}_{12}\hat{D}_{21}\geq0\label{cond2}\\
&&\hat{D}_{11}J_{22}+\hat{D}_{22}J_{11}-\hat{D}_{12}J_{21}-\hat{D}_{21}J_{12}>0\label{cond3}\\
&&(\hat{D}_{11}J_{22}+\hat{D}_{22}J_{11}-\hat{D}_{12}J_{21}-\hat{D}_{21}J_{12})^2-4\det\hat{D}\det J\geq0\label{cond4}
\end{eqnarray}
where the diffusion matrix is given by
\begin{equation*}
\hat{D}=\left(\begin{array}{cc}\hat{D}_{11} & \hat{D}_{12}\\ \hat{D}_{21} & \hat{D}_{22}\end{array}\right)=\left(\begin{array}{cc}D_{S}&c\frac{K (r-(\beta -\gamma))}{r}\\0&D_{I}\end{array}\right).
\end{equation*}
Notice that in the absence of cross-diffusion, $D_{12}=D_{21}=0$, Conditions \eqref{condJ1}, \eqref{condJ2}, and \eqref{cond4} become Conditions \eqref{Turing_no_cross-diffusion1}, \eqref{Turing_no_cross-diffusion2}, and \eqref{Turing_no_cross-diffusion3} from Model \eqref{stage}.

Conditions \eqref{condJ1} and \eqref{condJ2} hold if
\begin{equation}
\beta<r+\gamma,
\end{equation}
which is equivalent to
\begin{equation}
\mathcal{R}_0<\frac{r}{\gamma}+1.
\end{equation}
thus we must have that
\begin{equation}
1<\mathcal{R}_0<\frac{r}{\gamma}+1.
\end{equation}
It can be shown that Conditions \eqref{cond1} and \eqref{cond2} hold if $D_{S},D_{I}\neq0$, while Condition \eqref{cond3} holds if
\begin{equation}
Z:=-D_{S}\frac{(\beta-\gamma)}{\mathcal{R}_0}-D_{I}r+D_{I}\beta-D_{I}\frac{\gamma}{\mathcal{R}_0}-\frac{(\beta-\gamma)^2}{\beta}c\frac{K(r-(\beta-\gamma))}{r}>0
\end{equation}
and Condition \eqref{cond4} holds if
\begin{equation}
Z^2\geq D_{S}D_{I}\gamma(\beta-\gamma)(r-\beta+\gamma)
\end{equation}
\qed
\end{proof}

When $\alpha_1=0$, constant recruitment, System \eqref{general system} becomes
\begin{equation}
\begin{aligned}
\frac{\partial S}{\partial t}&=\Lambda-\beta \frac{SI}{S+I}+D_{S}\nabla^2 S+c\nabla\cdot(S\nabla I)\\
\frac{\partial I}{\partial t}&=\beta \frac{SI}{S+I}-\gamma I+D_{I}\nabla^2 I
\end{aligned}
\label{constant recruitment}
\end{equation}
where $\Lambda=r\left(1-\frac{1}{K}\right)$.

\begin{lemma}
Model \eqref{constant recruitment} does not support Turing's diffusive instability.
\end{lemma}
\begin{proof}
\smartqed
The endemic equilibrium is
\begin{equation}
(S^{*},I^{*})=\left(\frac{\Lambda }{\beta -\gamma },\frac{\Lambda}{\gamma}\right).
\end{equation}
The basic reproductive number is
\begin{equation}
\mathcal{R}_0=\frac{\beta}{\gamma}.
\end{equation}
and we assume $\beta>\gamma$ so that $\mathcal{R}_0>1$.

The Jacobian is
\begin{equation}
J=\left(\begin{array}{cc}
-\frac{(\beta-\gamma)^2}{\beta}	&	-\frac{\gamma^2}{\beta}\\
\frac{(\beta-\gamma)^2}{\beta}	&	-\frac{\gamma(\beta-\gamma)}{\beta}
\end{array}\right),
\end{equation}
and the diffusion matrix is
\begin{equation*}
\hat{D}=\left(\begin{array}{cc}D_{S}&c\frac{\Lambda}{\beta -\gamma}\\0&D_{I}\end{array}\right).
\end{equation*}
Note that Condition \eqref{cond3} fails due to the assumption that $\beta>\gamma$, and thus from \cite{Kumar2011} and \cite{Zemskov2013} we know that Turing diffusive instability is not possible.
\qed
\end{proof}
In short, logistic recruitment seems critical for supporting Turing's diffusive instability in the proposed cross-diffusion model.

\subsection{Effects of incidence functions}

The literature has often focused on modeling epidemics using the so called ``mass-action" law ($\alpha_2=0$) or ``standard" incidence ($\alpha_2=1$). In this section, we explore the role of this assumption in support of diffusive instability in our setting.

When $\alpha_2=0$, the mass action law comes into play and System \eqref{general system} becomes
\begin{equation}
\begin{aligned}
\frac{\partial S}{\partial t}&=rS\left(1-\frac{S}{K}\right)-\beta SI+D_{S}\nabla^2 S+c\nabla\cdot(S\nabla I)\\
\frac{\partial I}{\partial t}&=\beta SI-\gamma I+D_{I}\nabla^2 I
\end{aligned}
\label{mass action}
\end{equation}
\begin{lemma}
System \eqref{mass action} will not support Turing's diffusive instability.
\end{lemma}
\begin{proof}
\smartqed
The endemic equilibrium is
\begin{equation}
(S^{*},I^{*})=\left(\frac{\gamma}{\beta },\frac{r(K\beta-\gamma)}{K\beta^2}\right),
\end{equation}
where the basic reproductive number is
\begin{equation}
\mathcal{R}_0=\frac{\beta K}{\gamma}.
\end{equation}
The Jacobian evaluated at the endemic equilibrium is
\begin{equation}
J=\left(\begin{array}{cc}
-\frac{r}{\mathcal{R}_0}	&	-\gamma\\
r\left(1-\frac{1}{\mathcal{R}_0}\right)	&	0
\end{array}\right).
\end{equation}
The diffusion matrix is
\begin{equation*}
\hat{D}=\left(\begin{array}{cc}D_{S}&c\frac{\gamma}{\beta}\\0&D_{I}\end{array}\right).
\end{equation*}
Notice that Condition \eqref{cond3} fails if $\mathcal{R}_0>1$ is imposed. Thus \eqref{mass action} will not result in Turing's diffusive instability.
\qed
\end{proof}
The case $\alpha_2=1$, standard incidence, corresponds to the case when System \eqref{general system} becomes System \eqref{BerresModel}, and so, Turing's diffusive instability is possible. The use of standard incidence seems critical to the support of Turing's diffusive instability in our setting.

\subsection{Necessary and sufficient conditions}
\begin{theorem}
For a density dependent cross-diffusion SI model of the form System \eqref{general system}, logistic recruitment and standard incidence functions are necessary for Turing's diffusive instability.
\end{theorem}
\begin{proof}
\smartqed
The proof is a direct result of the preceding lemmas.
\qed
\end{proof}

See Figure \ref{cross-diffusion figure} for simulations for Model \eqref{BerresModel} carried out under Conditions \eqref{cross-diffusion cond1}-\eqref{cross-diffusion cond3}.

\begin{figure}
 \begin{tabular}{cc}
   \includegraphics[width=3.0in]{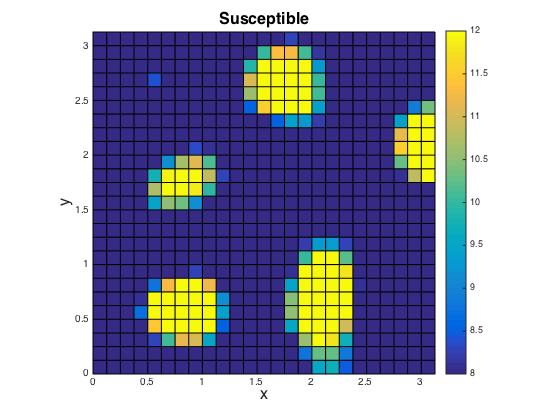}
   \includegraphics[width=3.0in]{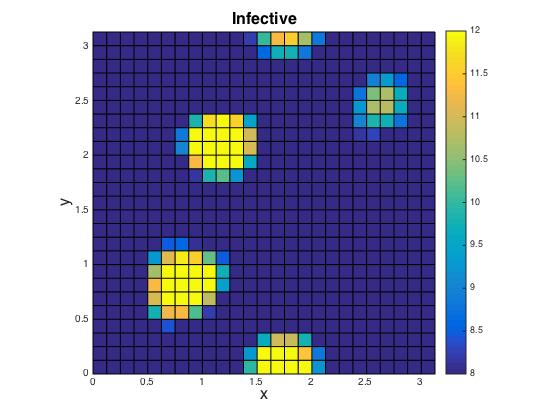}\\
   \includegraphics[width=3.0in]{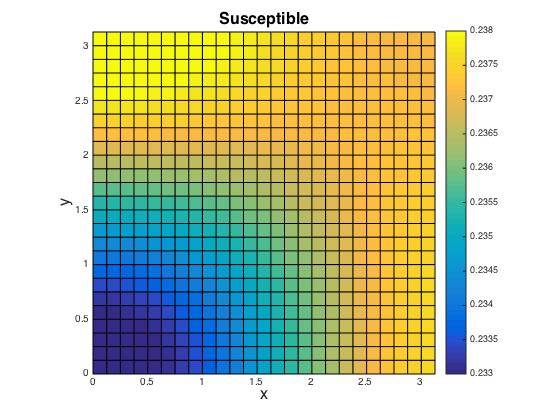}
   \includegraphics[width=3.0in]{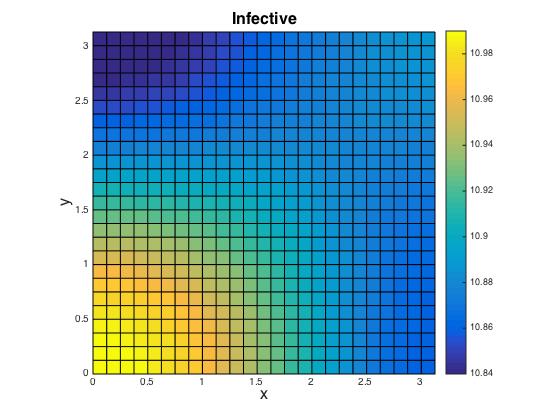}
 \end{tabular}
 \label{cross-diffusion figure}
 \caption{The distribution for Model \eqref{general system} under constant recruitment ($\alpha_1=0$) and mass action incidence function ($ \alpha_2=0$) under the conditions $D_S=0.1, D_I=2, c=0.02, r=0.4, K=100, \beta=0.5$, chosen so that Conditions \eqref{cross-diffusion cond1}-\eqref{cross-diffusion cond3} are satisfied for $t=0$ (top), $t=500$ (bottom). As time increases the distributions of both susceptible and infective populations have a homogeneous distribution, with no patches.}
\end{figure}

\section{Discussion and Conclusion}\label{Discussion and Conclusion}

We have proposed two models: a phenomenological model that examined the effects of an ``unusual" incidence function and a cross-diffusion model. Model (\ref{stage}) can be applied to the study of sexually transmitted diseases such as \emph{chlamydia} and \emph{gonorrhea} as well as to communicable diseases like \emph{leprosy} or possibly Ebola. In all three examples some form of \emph{social distancing} is assumed to be generated in response to the presence of symptoms. Model \eqref{stage} predicts that changes in behavior will result in spatial aggregation (via diffusive instability) and that such natural responses help, in fact, to reduce the population's levels of infection.

An $SI$ model with density dependent cross-diffusion, where susceptible individuals avoid increasing gradients of infective individuals is also considered. Indeed if the sign of the diffusion coefficient is negated, individuals would be attracted to increasing gradients of infective populations, as shown in the Keller-Segel model \cite{Keller1970,Keller1971a}, rather than repelled from infective populations. Using Model \eqref{general system} as a starting point, we examine the effects of the choice of recruitment and incidence functions and conclude that a logistic recruitment and standard incidence functions are necessary to have pattern formations, Turing's diffusive instability. Mass action incidence function, a popular choice, does not result in diffusive instability.

\section{Appendix: Derivation of the $SI$ model with diffusion}\label{derivation}
As a starting point, let the densities for populations susceptible to a disease and infective with a disease, at time $t$ and position $(x,y)\in\Omega$ be $S(x,y,t)$ and $I(x,y,t)$, respectively. We assume the model takes the form of the following reaction-diffusion model
\begin{equation}
\begin{aligned}
\frac{\partial S}{\partial t}&= -\nabla\cdot\mathbf{J}_1+f_1(S,I) & \mbox{in}~ \Omega \times (0,T),\\
\frac{\partial I}{\partial t}&=-\nabla\cdot\mathbf{J}_2 +f_2(S,I)& \mbox{in}~ \Omega \times (0,T),
\label{sys1}
\end{aligned}
\end{equation}
where $f_1, f_2$ and $\mathbf{J}_1, \mathbf{J}_2$ are the reaction and flux terms for the susceptible and infective populations, respectively. The reaction terms are modeled as follows:
\begin{equation*}
f_1(S,I)=rS^{\alpha_1}\left(1-\frac{S^{\alpha_1}}{K}\right)-\beta \frac{SI}{(S+I)^{\alpha_2}}, \qquad
f_2(S,I)=\beta \frac{SI}{(S+I)^{\alpha_2}}-\gamma I,
\end{equation*}
where $K$ is the carrying capacity, $r$ is the intrinsic birth rate, $\beta$ is the transmission rate, $\gamma$ is the recovery rate, $\alpha_1\in\{0,1\}$, $\alpha_2\in\{0,1\}$; $\alpha_1=1$ corresponding to logistic growth and $\alpha_1=0$ to constant recruitment; $\alpha_2=0$ accounts for mass-action transmission while $\alpha_2=1$ models standard incidence.

It is assumed that each population is influenced by increasing gradients of infectious individuals that result in the ``directional" dispersive migrations of each population towards its own type. Let $D_{S}$ and $D_{I}$ be the intrinsic-diffusion constants of the susceptible and infective populations, respectively, then the intrinsic dispersal forces of $S$ and $I$ in the flux are given by the gradient of the densities, $D_{S}\nabla S$, $D_{I}\nabla I$, respectively \cite{Nallaswamy1982}. The assumption that $D_{S},D_{I}\geq0$ means that the dispersal is in directions away from high densities, the last assumption justified by the tendency of susceptible to avoid increasing gradient populations of symptomatic infectious individuals, that is, it is assumed that they tend to move towards decreasing gradients of symptomatic individuals. The cross-diffusion coefficient measuring the impact of the infective population on the susceptible population is denoted by the constant $c\geq0$. Therefore, the cross-diffusion force of infective on susceptible populations in the flux is given by $cS\nabla I$. We further assume that there are no other cross-diffusion forces. Thus the flux for $S$ and $I$ are modeled as
\begin{equation*}
\begin{aligned}
\mathbf{J}_1&=-D_{S}\nabla S - cS\nabla I\\
\mathbf{J}_2&=-D_{I}\nabla I,
\end{aligned}
\end{equation*}
which takes the form of the celebrated Keller-Segel model \cite{Keller1970,Keller1971a}.

Incorporating the reaction and diffusion terms leads to the following $SI$ epidemiological system
\begin{equation}
\begin{aligned}
\frac{\partial S}{\partial t}&=rS^{\alpha_1}\left(1-\frac{S^{\alpha_1}}{K}\right)-\beta \frac{SI}{(S+I)^{\alpha_2}}+D_{S}\nabla^2 S+c\nabla\cdot(S\nabla I)\\
\frac{\partial I}{\partial t}&=\beta \frac{SI}{(S+I)^{\alpha_2}}-\gamma I+D_{I}\nabla^2 I.
\end{aligned}
\end{equation}
Finally, it is further assumed that we have a closed system involving no external input; thus the use of Neumann boundary conditions
\begin{equation}
\frac{\partial }{\partial {\nu}}S=\frac{\partial }{\partial {\nu}}I=0,
\end{equation}
is acceptable. The initial conditions are as follows
\begin{equation}
S(x,y,0)=S_0(x,y) \qquad \mbox{and} \qquad I(x,y,0)=I_0(x,y).
\end{equation}

Whenever $D_{S}$ and $D_{I}$ are the dominant coefficients, System \eqref{general system} reduces essentially to the heat equation, which under Neumann boundary conditions will go to the average of the initial data as $t\rightarrow\infty$ \cite{Ni1998}.

\end{document}